\documentclass[aps,prd,showpacs,onecolumn,nofootinbib,superscriptaddress]{revtex4-2}
\usepackage{amsmath,amssymb,amsfonts}
\usepackage{bm}
\usepackage{verbatim}
\usepackage{hyperref}
\usepackage{slashed,braket}
\usepackage{ulem}
\usepackage{cancel}
\usepackage{color}
\usepackage{graphicx,graphics}
\usepackage[title,titletoc]{appendix}
\usepackage{stackengine}
\usepackage{mathrsfs}
\usepackage{tikz-feynman}
\usepackage{amsthm}

\usepackage{tikz}
\usepackage{pgfplots}
\usepackage{lineno}
\usepackage[T1]{fontenc} % if needed
\usepackage{braket}

\usepackage{bbold}
\usepackage{comment}
\excludecomment{confidential}
\usepackage{makecell}
\usepackage{bm}
\usepackage{braket}
\usepackage{slashed}
\usepackage{mathtools}

\DeclareMathOperator{\sgn}{sgn}
\usepackage{amsthm}

\newtheorem{theorem}{Theorem}

\newcommand{\mza}[1]{{#1}}

\newcommand{\mzo}[1]{{#1}}
\newcommand{\pd}[1]{{#1}}
\newcommand{\mzc}[1]{{#1}}
\newcommand{\mzz}[1]{{#1}}
\newcommand{\mzu}[1]{{#1}}
\newcommand{\Ll}[1]{\overleftarrow{#1}}
\newcommand{\Rr}[1]{\overrightarrow{#1}}

\newcommand{\W}[1]{\widetilde{#1}}

\newcommand{\Tr}{\operatorname{Tr}}

\newcommand{\tr}{\operatorname{tr}}

\begin{document}
	
	\title{Chiral anomaly in inhomogeneous systems with nontrivial momentum space topology}
	
	\author{Praveen D. Xavier}
	\email{praveen.xavier@msmail.ariel.ac.il}
		\affiliation{Physics Department, Ariel University, Ariel 40700, Israel}
	
	\author{M.A.Zubkov}
	\email{mikhailzu@ariel.ac.il}
	\affiliation{Physics Department, Ariel University, Ariel 40700, Israel}
	
	\date{\today}
	
	\begin{abstract}
		We consider the chiral anomaly for systems with a wide class of Hermitian Dirac operators ${Q}$ in 4D Euclidean spacetime. We suppose that $ Q$ is not necessarily linear in derivatives and also that it contains a coordinate inhomogeneity unrelated to that of the external gauge field. We use the covariant Wigner-Weyl calculus (in which the Wigner transformed two point Greens function \mzo{belongs to the two-index tensor representation of the gauge group}) and point splitting regularization to calculate the global expression for the anomaly. The Atiyah-Singer theorem can be applied to relate the anomaly to the topological index of $ Q$. We show that the topological index factorizes (under certain assumptions) into the topological invariant $\frac{1}{8\pi^2}\int \text{tr}(F\wedge F)$ (composed of the gauge field strength) multiplied by a topological invariant $N_3$ in phase space. The latter is responsible for the topological stability of Fermi points/Fermi surfaces and is related to the conductivity of the chiral separation effect. 
	\end{abstract}
	\pacs{}
	
	\maketitle
	%\tableofcontents

\section{Introduction}

The chiral anomaly is by now a famous result in quantum field theory \cite{adler1969axial, bell1969pcac, bertlmann2000anomalies, fujikawa1979path, fujikawa2004path}. In a system with chiral symmetry, and an associated conserved chiral current \mzo{on the level of classical Lagrangian}, it measures the degree to which the quantum expectation value violates the conservation law. It was first discovered in quantum electrodynamics \cite{adler1969axial, bell1969pcac}; and quickly extended to Yang-Mills theories \cite{bardeen1969anomalous}. \mzo{The} \pd{integrated (i.e. global)} anomaly turned out to be a topological invariant of the gauge field configuration -- namely the second Chern class of the bundle \cite{alvarez1984topological}. This peculiarity found its explanation through the Atiyah-Singer (AS) index theorem \cite{atiyah1963index} (a mathematical theorem predating the discovery of the anomaly by around 5 years). 
The AS theorem, simply stated, equates the ``algebraic index'' of an elliptic differential operator to its ``topological index'' (which in turn has an expression in terms of an integral of a density over the ``phase space'' \cite{fedosov1996deformation}). 
\pd{The connection of this theorem with the anomaly comes from the method of Fujikawa \cite{fujikawa1979path,fujikawa1980path} who revealed that the the global anomaly is proportional to the algebraic index of the differential operator
governing the dynamics of the Dirac fermion.}
If we assume the operator is elliptic, the AS theorem is applicable, and the anomaly can be equated to the topological index of the operator. 
Due to the generality of the AS theorem and Fujikawa method, though, one can assert that the \mzo{global} chiral anomaly arising from \emph{any} elliptic operator is equal to a topological invariant. 
Let us consider this in more detail.

In $D=4$ Euclidean spacetime dimensions consider a Dirac fermion $\psi$ in the fundamental representation of some gauge group $G$ with Lie algebra $\mathfrak g$; and let it be governed by the partition function
\mzc{\begin{align}
Z=&\int D\bar\psi D\psi \,e^{\int d^4x \,\bar\psi(x){Q}\psi(x)}\\
\text{where}\quad  Q=& \left(\begin{array}{cc}
	0 &   {O}^\dagger\\  {O} & 0 
\end{array} \right) \label{Q}
\end{align}}
and $ O$ is a \mzo{matrix} of differential operators which we suppose have the following form \cite{bonora2023fermions} (for some positive integer $m$):
\begin{equation}
 {O} = \sum_{|\alpha| \leq m} f_{\alpha}(x) (-i\partial)^\alpha \label{Oop}
\end{equation}
with the multi-index $\alpha=(\alpha_1,\alpha_2,\alpha_3,\alpha_4)$, $|\alpha| := \sum_\mu \alpha_\mu$ and 
$(-i\partial)^\alpha := \prod_\mu (-i\partial_\mu)^{\alpha_\mu}$.
Here $f_{\alpha}(x)$ is a \mzo{ matrix-valued function of the coordinates. \pd{(}This is a matrix in both spinor \pd{(2x2)} and internal spaces. The internal space contains both a fundamental representation of $G$ and another internal space.\pd{)}} 
%where $ O$ is $2\times 2$ and a differential operator. We assume it has to the form
By the `\mzo{principal} symbol' of ${O}$ we will denote the matrix-valued function of $x$ and $p$ \cite{bonora2023fermions}:
\begin{equation}
o(x,p) := \sum_{|\alpha| = m} f_{\alpha}(x) p^\alpha
\end{equation} 
Let us assume that $ O$ is elliptic: i.e. the matrix $o(x,p)$ is invertible when $p \ne 0$ \cite{freed2021atiyah_arxiv}. It is then also a Fredholm operator \cite{fedosov1996deformation} (Prop. 4.2.2), i.e.  ${\rm dim}\,{\rm ker}\,{O}$ and ${\rm dim}\,{\rm ker}\,{O}^\dagger$  are finite.
The algebraic index of ${O}$ is defined as the difference:
\begin{equation}
    {\rm index}\,{O} := {\rm dim}\,{\rm ker}\,{O} - {\rm dim}\,{\rm ker}\,{O}^\dagger 
\end{equation} 
In the chiral basis $\gamma_5=({\bm 1}_2,-{\bm 1}_2)$. It is easy to see then that \cite{nakahara2003geometry}
\mzc{\begin{equation}
    n_+-n_-={\rm dim}\,{\rm ker}\,{O}-{\rm dim}\,{\rm ker}\,{O}^\dagger ={\rm index}\, O \label{ind}
\end{equation}}
where $n_+$ (resp. $n_-$) is defined as the number of zero modes of $ Q$ with positive (resp. negative) chirality (i.e. eigenvalues under multiplication by $\gamma_5$ are $\pm 1$ respectively). It is easy also to see that $\{\gamma_5,Q\}=0$, i.e. chiral transformations (in which $\psi\to e^{i\alpha\gamma_5}\psi$, $\alpha\in\mathfrak g$) are a symmetry of the action. Let us denote the corresponding Noether current of the symmetry by $J_\mu$, which lies in the \mzo{two-index reducible representation of $G$} and which satisfies the \emph{classical} conservation law $\mathcal D_\mu J_\mu=0$ where $\mathcal D=\partial-i[A,]$ is the gauge covariant derivative acting in the above \mzo{mentioned} representation. However, the \emph{quantum} expectation value $\langle \mathcal \int \tr \mathcal {D}_\mu J_\mu\rangle$ doesn't necessarily have to vanish. In fact a short argument reveals that \cite{nakahara2003geometry}
\begin{equation}
    \mathscr A:=\mzo{}\int\langle \mathcal  \tr \mathcal{D}_\mu J_\mu\rangle=\mzo{}2i(n_+-n_-) \label{??}
\end{equation}
($\mathscr A$ is the anomaly.)
%where $n_+$ (resp. $n_-$) as the number of zero modes of $ Q$ with positive (resp. negative) chirality (i.e. eigenvalues under multiplication by $\gamma_5$ are $\pm 1$ respectively). 
The argument goes as follows:
under a change of variables corresponding to a chiral transformation $\psi\to\psi'= e^{i\alpha\gamma^5}\psi$, $\bar\psi\to\bar\psi'=\bar\psi e^{i\alpha\gamma^5}$, $\alpha\in \pd{ \mathbf R}$, the partition function $Z$ must be unchanged. There are two terms which contribute to the variation of $Z$: one from the change of measure \cite{fujikawa1979path} and another from the action. The measure transforms as $D\psi D\bar\psi\to D\psi D\bar\psi\det^{-1}(\delta \psi'/\delta\psi)\det^{-1}(\delta\bar\psi'/\delta\bar\psi)=D\psi D\bar\psi\exp(-2i\alpha\text{Tr}(\gamma_5))$ where we've used the formula $\det=\exp\Tr\log$, with $\text{Tr}$ standing for the trace over eigenstates of $Q$; and the action transforms as $S\to S+\alpha\int \text{tr}\mathcal D_\mu J_\mu+O(\alpha^2)$.
So
\begin{equation}
    \frac{d Z}{d\alpha}\bigg{|}_{\alpha=0}\overset{!}{=}0=-2i\text{Tr}(\gamma_5)+\int \left\langle \tr \mathcal{D}_\mu J_\mu\right\rangle
\end{equation}
Since $\{ Q,\gamma_5\}=0$, $\gamma_5$ sends an eigenstate of $Q$ with eigenvalue $\lambda$ to an eigenstate with eigenvalue $-\lambda$. And since $ Q$ is Hermitian, eigenstates of different eigenvalue are orthogonal. Therefore, what remains in $\text{Tr}(\gamma_5)$ is only the trace over the zero modes of $ Q$. Those zero modes can be classified by their chirality -- so finally $\text{Tr}(\gamma_5)=n_+-n_-$ which proves \eqref{??}. 
%Combining this with \eqref{ind} we find
%\begin{equation}
%    \int\langle \mathcal  \tr D_\mu J_\mu\rangle=2i\,{\rm index}\, O 
%\end{equation}

Now, a version of the Atiyah-Singer index theorem for elliptic operators over Euclidean space \cite{fedosov1996deformation} (Prop. 4.2.8) gives us
\begin{equation}
    {\rm index}\, O=\int d^4x d^4p \,\text{ch}(\xi)(x,p)=\text{topological index } O \label{AS}
\end{equation}
where $\text{ch}(\xi)$ is the Chern character of the associated ``virtual bundle'' $\xi$ (see \cite{fedosov1996deformation} for the full details.) The RHS is the ``topological index'' of $ O$.
%where the integrand is the Chern character of the symbol $o(x,p)$ -- it reflects properties of the homotopic classes of the function $o(x,p)$ and it is usually introduced via K-theory \cite{}. 
Combining \eqref{ind}, \eqref{??} and \eqref{AS} we obtain
\begin{equation}
    \mathscr A=\mzo{2i}\int d^4x d^4p \,\text{ch}(\xi)(x,p)=\mzo{2i}\times\text{topological index } O\label{AS2}
\end{equation}
When $ Q$ is the conventional Dirac operator $ Q=i\gamma_\mu D_\mu$ then the Chern character reduces to the second Chern class $\pd{-}\frac{1}{8\pi^2}\int\text{tr}F\wedge F$ of the given principle G-bundle, giving the conventional form of the anomaly: 
%and \eqref{AS2} becomes
\begin{equation}
	\mathscr A= \mzo{-\frac{i}{4\pi^2}}\int\text{tr}F\wedge F\label{AS3}	
\end{equation}
%In this paper, however, we will be interested in calculating the anomaly for general elliptic operators; the method we will be using is the ``covariant Wigner-Weyl calculus'' (developed by us) and point-splitting regularization to deal with divergences.
In this paper we will consider a wide class of $ Q$. We assume that the gauge field is minimally coupled to the derivative everywhere, so that $Q$ can be written as
\begin{equation}
    Q=\sum_{|\alpha|\leq m}  c_\alpha(x) (-iD)^\alpha
\end{equation}
where $D:=\partial-iA$ and $c_\alpha(x)$ is a \mzo{$4\times 4$ matrix of functions, which are themselves matrices in internal space, but are singlets in $G$.} 
\pd{(We will require that $Q$ is Hermitian and $\{\gamma_5,Q\}=0$ which implies that it has the form \eqref{Q} so that the previous analysis is applicable.)}
The only \pd{significant} restriction we put on $Q$ is that $Q^{(A=0)}=\sum_{|\alpha|\leq m}  c_\alpha(x) (-i\partial)^\alpha$ (i.e. with the gauge field set to zero) has a (conventional) ``Wigner transform'' \cite{zachos2005quantum,chernodub2017scale,zhang2020influence} $Q_W(x,p)$ 
%if we set the gauge field to zero in $Q$ and consider its ``Wigner transform'' $Q^{(A=0)}_W(x,p)$ that it is 
homotopic to a function of only $p$. (This hypothesis captures our essential assumption that topology in coordinate space arises only from the external gauge field configuration.)
Under this hypothesis we will show that 
\begin{align}
    \mathscr A=& \mzo{-N_3}\times \frac{i}{4\pi^2}\int\text{tr}F\wedge F\\
    \text{where}\quad N_3:=&\mzo{\frac{1}{48\pi^2|V|}}\int d^3\vec x\int _\Sigma\text{tr}_D\left(\gamma^5  G^{(0)}\star d Q_W\star  G^{(0)}\star\wedge d  Q_W \star G^{(0)}\star\wedge d  Q_W\right)
\end{align}
%\sum_{p_4=0^{\pm}}\frac{\sgn(p_4)}{48\pi^2V}\int d^3x\int d^3\vec p\,\varepsilon_{ijk}\,\text{tr}\left(\gamma^5 G^{(0)}\star\partial_{p_i} Q\star G^{(0)}\star\partial_{p_j} Q\star G^{(0)}\star\partial_{p_k} Q\right)
(a neat factorization). Here the 3-surface $\Sigma$ defined as the union of the two hyperplanes $p_4=0^{\pm}$, while $G^{(0)}(x,p)$ is defined as solution of equation $G^{(0)}\star Q_W = 1$ with standard Moyal product $\star$. $N_3$ is responsible for the stability of Fermi surfaces/Fermi points \cite{volovik2003universe}. In parallel $N_3$ enters the expression for the conductivity of the \mzo{chiral separation effect (CSE)} \cite{zubkov2023effect}. 
%In order to arrive at this factorization we require that ${Q}$ does not carry topology in coordinate space if the external field $A$ is removed. Namely, we formalize this condition as follows. Let us consider the Weyl symbol $Q(x,p)$ of ${Q}$ with vanishing external gauge field $A$. We require that $Q(x,p)$ is homotopic to a function $Q(p)$ that does not depend on coordinates. 
%In other words we calculate here the topological index of generalized Dirac operators of rather general form with both nontrivial internal momentum space topology (provided by nonlinear dependence on momenta) and the essential inhomogeneity in coordinate space (which takes place even without external gauge field). 

We calculate the anomaly using the ``covariant Wigner-Weyl calculus'' method developed by us earlier in \cite{zubkovxavier} and point splitting regularization \cite{schwinger1951gauge, peskin1995introduction, bertlmann2000anomalies, ioffe2006axial} to regulate ultraviolet divergencies.

\section{Covariant Wigner-Weyl calculus}
In a previous paper \cite{zubkovxavier} we developed the machinery of ``covariant Wigner-Weyl calculus''. Below is a summary and introduction to the method:
\begin{itemize}
    \item Let us work in 4 Euclidean spacetime dimensions. Let $A_\mu(x)$ be an external gauge field of a matrix Lie group $G\subset GL(N,\bold C)$ with Lie algebra $\mathfrak g$. Let us introduce a Hilbert space, $H$, of a suitable space of functions over $\bold R^4$. And let us use the ``bra-ket'' notation \cite{griffiths2018introduction} for states in $H$. The ``position and momentum operators'' will be denoted by $\hat x_\mu$ and $\hat p_\mu$ respectively (i.e. $\hat x_\mu$ corresponds to $\psi(x)\to x_\mu\psi(x)$ and $\hat p_\mu$ corresponds to $\psi(x)\to-i\partial_\mu \psi(x)$). \mzo{We suppose that in addition to the space of the gauge group (i.e. of its fundamental representation) there is also another internal space of dimension $M$. Let $\hat X$ be a $(N\times M\times 4)\times (N\times M \times 4)$ matrix of operators in the Hilbert space ($4$ is the dimension of Dirac spinors).} Then the ``covariant Wigner transform'' (hereafter simply the ``Wigner transform'') of $\hat X$ is defined as
\begin{align}
	  X_{{W}}(x,p):=&\int d^4y\,e^{ipy}\,U(x , x-y/2)\bra{x-y/2}\hat{X}\ket{x+y/2}  U(x+y/2 , x) \label{12fe}\\        \text{where}\quad U(y,x)=&\text{Pexp}\left(i\int_{x\to y} dz^\mu \,  A_\mu(z)\right)
\end{align}
$U(y,x)$ is the path-ordered exponential along the straight path from $x$ to $y$. Definition \eqref{12fe} first appeared in \cite{vasak1987quantum} for Abelian gauge fields.
\newline It is ``covariant'' in the sense that $X_W(x,p)$ \mzo{belongs to the two-index reducible tensor representation of $G$, and it} transforms as in the adjoint representation under gauge transformations: $X_W(x,p)\to \Omega(x) X_{{W}}(x,p) \Omega(x)^\dagger$. 

    \item \eqref{12fe} can also be written as
    \begin{align}
      X_{{W}}(x,p)&=\int d^4y\,e^{ipy}\,\bra{x}e^{-\frac{i}{2}y\hat\pi}\hat X \,e^{-\frac{i}{2}y\hat\pi}\ket{x} \\
    \text{where}\quad \hat\pi_\mu&:=\hat p_\mu-A_{\mu}(\hat x) \label{pi}
    \end{align}
    This follows from the following identities:
    \begin{align}
    e^{-iy\hat{ {\pi}}}\ket{x}&=\ket{x+y} {U(x+y,x)}\label{usin1}\\
    \text{and}\quad\bra{x}	e^{iy\hat{{\pi}}}&={  U(x , x+y)}\bra{x+y} \label{usin5}
    \end{align}
    which in turn follow from the following functional identities:
    \begin{align}
    \exp(y D)\psi(x)
    &=U(x,x+y)\psi(x+y)\\
    \exp(y \mathcal D)\Psi(x)
    &=U(x,x+y)\Psi(x+y)U(x+y,x) \label{adjoint}
    \end{align}
    where $\psi$/$\Psi$ is in the fundamental/adjoint representation and $D:=\partial-iA$ and $\mathcal D:=\partial-i[A,]$.
    
    \item The inverse (Weyl) transform of \eqref{12fe} is
\begin{equation}
    \hat X=(2\pi)^{-8}\int d^4q d^4y d^4x d^4p \,e^{\frac{i}{2}y(\hat{ {\pi}}-p)}e^{iq(\hat x-x)} {X}_{{W}}(x,p)\,e^{\frac{i}{2}y(\hat{ {\pi}}-p)} \label{2}
\end{equation}

%\item It is also useful to note the following identities:
%\begin{align}
%    \exp(y D)\psi(x)
%    &=U(x,x+y)\psi(x+y)\\
%    \exp(y D)\psi(x)
%    &=U(x,x+y)\psi(x+y)U(x+y,x)
%\end{align}
%for $\psi$ in the fundamental and adjoint representations receptively. (We've defined our covariant derivative as $D_\mu\equiv \partial_\mu-iA_\mu(x)$.)

\item For two operators $\hat X$ and $\hat Y$, we define the star product as $X_W\bigstar Y_W:=(\hat X\hat Y)_W$. The formula for it is 
\begin{align}
\begin{split}
    &(X_W\bigstar Y_W)(x,p)\\
    =&(2\pi)^{-8}\int d^4yd^4k d^4y'd^4k'\,e^{-iy(k-p)-iy'(k'-p)}\times \\
    &U(x , x-(y+y')/2) \, U(x-(y+y')/2 , x-y'/2) \, X_W(x-y'/2,k) \, U(x-y'/2 , x+(y-y')/2)\\
    & U(x+(y-y')/2, x+y/2) \, Y_W(x+y/2,k') \, U(x+y/2,  x+(y+y')/2) \, U(x+(y+y')/2 , x)
\end{split} \label{starer}
\end{align}
The integrand is the triangular Wilson loop from $x\to x+(y+y')/2\to x+(y-y')/2\to x-(y+y')/2\to x$ with $X_W$ and $Y_W$ inserted at $x-y'/2$ and $x+y/2$, respectively, on the triangle -- see Fig. \ref{mfig}.
\begin{figure}
    \centering
    \includegraphics[scale=0.3]{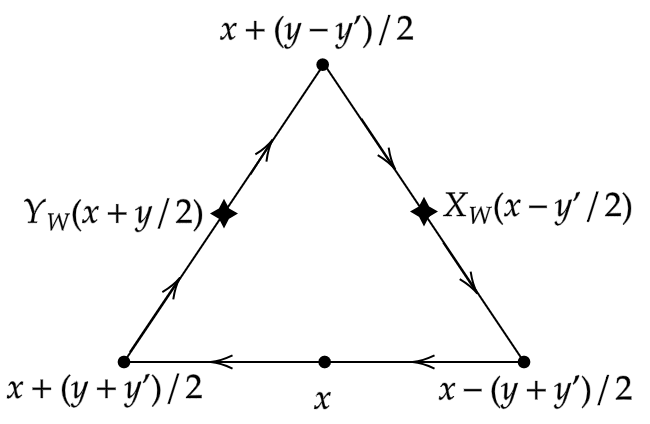}
    \caption{The meaning of integrand \eqref{starer}.}
    \label{mfig}
\end{figure}

In comparison, the ``Moyal star product'' \cite{hillery1984distribution, moyal1949quantum,zachos2005quantum} for two functions $f, g$ on the phase-space with trivial gauge indices (i.e. proportional to the identity) is
\begin{equation}
    (f\star g)(x,p):=(2\pi)^{-8}\int d^4yd^4k d^4y'd^4k'\,e^{-iy(k-p)-iy'(k'-p)}
		 f(x-y'/2,k) 
		  g(x+y/2,k') \label{moyalstar}
\end{equation}

\item Denoting trace over the gauge indices by $\text{tr}_G$ and trace over the Hilbert space by $\text{tr}_H$ we have
%\begin{align}
%    {\rm tr}_G {\rm tr}_H (\hat{X} \hat{Y}) =& (2\pi)^{-4}\int d^4x d^4p \,{\rm tr}_G (X_W \bigstar Y_W) = (2\pi)^{-4}\int d^4x d^4p \,{\rm tr}_G (X_W  Y_W) \label{TrXY}
%\end{align}
\begin{align}
    {\rm tr}_G {\rm tr}_H (\hat{X} \hat{Y})= (2\pi)^{-4}\int d^4x d^4p \,{\rm tr}_G (X_W  Y_W) \label{TrXY}
\end{align}
\mzo{Also, in the following, we will denote by $\tr_D$ the trace over spinor indices and also over the $M$-dimensional extra
internal space (not related to the gauge group).}

%\item Suppose that operators $\hat X$ and $\hat Y$ are proportional to unity element of the given gauge group. Then we define the ordinary Moyal product as
%\begin{align}
%	\begin{split}
%		&(X_W\star Y_W)(x,p)\\
%		=&(2\pi)^{-2D}\int dydk dy'dk'\,e^{-iy(k-p)-iy'(k'-p)}
%		 X_W(x-y'/2,k) 
%		  Y_W(x+y/2,k') 
%	\end{split} \label{starMoyal}
%\end{align}

\item 
\mza{(\ref{starer}) allows us to derive the following identities:
\begin{equation}
\int \frac{d^4x d^4p}{(2\pi^4)}	{\rm tr_G}(X_W(x,p)\bigstar Y_W(x,p)) = \int \frac{d^4x d^4p}{(2\pi^4)}	{\rm tr}_G(X_W(x,p) Y_W(x,p)) = \int \frac{d^4x d^4p}{(2\pi^4)}	{\rm tr}_G(Y_W(x,p)\bigstar X_W(x,p)) \label{cyclTr}
\end{equation}
provided that these integrals are convergent.}

\end{itemize}

\section{The model}

\subsection{Action in terms of Wigner-Weyl calculus}

Let $\psi(x)$ denote a four-component spinor transforming in the fundamental representation of the gauge group. The partition function of the model is
\begin{align}
    Z=&\int D\bar\psi D\psi\,e^S \label{pf} \\
    \text{with}\quad S=&\int d^4x \,\bar\psi(x)Q(x,-iD)\psi(x)\label{ac}\\
    \text{and}\quad Q(x,-iD)=&\sum_{|\alpha|\leq m}  c_\alpha(x) (-iD)^\alpha \label{sym}
\end{align}
\eqref{sym} is written in the multi-index notation of the introduction. $m$ is the \emph{order} of $Q$ and it is a positive integer. $D_\mu:=\partial_\mu-iA_\mu$ is the covariant derivative. $c_{\alpha}(x)$ is a \mzo{$(4M)\times (4M)$} matrix of functions.
%We assume that this inhomogeneity is time independent. 
We require that $Q^\dagger = Q$ and $\{\gamma_5, Q\}=0$, where $\gamma_5=(\bold 1_2,-\bold 1_2)$. This implies that $Q$ has the form $ Q=\begin{pmatrix}
    0&  O\\
    O^\dagger& 0
\end{pmatrix}$ as in \eqref{Q}. We also require that $O$ is elliptic as discussed in the introduction.
%At this stage we do not take into account interactions.
The ``inhomogeneity'' in this model comes from the explicit $x$ dependence in $Q(x,-iD)$ -- \emph{we will assume, though, that this doesn't include $x_4$}. %We will make the assumption that this in independent of $x_4$, the Euclidean time component.

In the Hilbert space operator notation \eqref{ac} can be written as 
\begin{align}
    S&=-\text{tr}_{D}\text{tr}_{ G}\text{tr}_{H}\left(\hat Q\hat\rho\right)\\
    \text{where}\quad\hat Q&:=Q(\hat x,\hat \pi)\\
    \text{and}\quad\bra{x}\hat{\rho}\ket{y}&:=\psi(x)\bar\psi(y) \label{usin}
\end{align}
where $\text{tr}_D$, $\text{tr}_G$ and $\text{tr}_H$ are traces w.r.t. \mzo{the $4M$ spinor and extra internal space indices, $N$ gauge indices and the
Hilbert space respectively.}

\emph{We note that $Q_W(x,p)$, the Wigner transform of $\hat Q$, is independent of the gauge field (therefore it is proportional to the identity matrix as far as gauge indices are concerned).} In fact $Q_W(x,p)$ is the \emph{ordinary Wigner transform} \cite{zachos2005quantum} of $\hat Q$ with the gauge field set to zero. That is 
\begin{equation}
    Q_W(x,p)=\int d^4y\,e^{ipy}\bra{x-y/2}\hat{Q}^{(A=0)}\ket{x+y/2} 
\end{equation}
where $\hat{Q}^{(A=0)}=\sum_{|\alpha|\leq m}  c_\alpha(\hat x) \hat p^\alpha$.
%One can verify this by using \eqref{2} to take the Weyl transform of $Q_W(x,p)$.

%and acts as the identity matrix on the fundamental representation.

%We note the following exact relationship:
%\begin{equation}
%    \left(\hat Q\right)_W(x,p) =Q(x,p) \label{exact}
%\end{equation} 
%which can be verified by taking the Weyl transform of $Q(x,p)$ using \eqref{2}. 
%This formula means that $Q_W(x,p)$ is not a matrix but a scalar function.

%and where $\hat Q:=Q(\hat x,\hat \pi)$ (recall $\hat \pi:=\hat p-A(\hat x)$) and $\bra{x}\hat{\rho}\ket{y}:=\psi(x)\bar\psi(y)$. 

\subsection{Regularization}

The trace of $\hat\rho$ involves the product of two fields at the same point: $\bra{x}\hat \rho\ket{x}=\psi(x)\bar\psi(x)$. This may lead to a short-distance singularity (i.e. an ultraviolet divergence) \cite{collins1984renormalization}. A method of regularizing this is ``point-splitting'' \cite{schwinger1951gauge, peskin1995introduction, bertlmann2000anomalies, ioffe2006axial}. In our context it is achieved by replacing $\hat \rho$ with $\hat \rho^\epsilon:=e^{i\hat{\pi}\epsilon}\hat\rho e^{i\hat{\pi}\epsilon}$. 
To see this, consider $\bra{x}\hat\rho^\epsilon\ket{x}=U(x,x+\epsilon)\psi(x+\epsilon)\bar\psi(x-\epsilon)U(x-\epsilon,x)$ (using \eqref{usin1}, \eqref{usin5} and \eqref{usin}). The fields are now separated by the small 4-vector $\epsilon$; and in addition there are Wilson lines connecting the disparate points to ensure that $\bra{x}\hat\rho^\epsilon\ket{x}$ transforms in the adjoint representation at $x$. Since $\hat \rho^{\epsilon=0}=\hat \rho$, taking the regulator $\epsilon\to 0$ at the end of the calculation should recover the original quantity.

\emph{We note that the Wigner transform of $e^{i\hat\pi\epsilon}$ is $e^{ip\epsilon}$.}

\subsection{Noether procedure in terms of Wigner-Weyl calculus}
\label{SectNoether}

In light of the regularization introduced above we replace \eqref{usin} by
\begin{align}
	S^\epsilon&=-\text{tr}_{D}\text{tr}_{ G}\text{tr}_{H}\left(\hat Q \hat\rho^\epsilon\right) \label{usinreg}
\end{align}
Since $\{\gamma^5,\hat Q\}=0$, $S^\epsilon$ is unchanged by a global non-Abelian chiral transformation: $\psi\to e^{i\alpha\gamma^5}\psi$, $\bar\psi \to \bar\psi e^{i\alpha\gamma^5}$ ($\alpha\in\mathfrak g$). The associated conserved Noether current can be extracted as follows: promote the global transformation to a local one: 
%The theory is to be regularized in order to provide finite values of physical quantities. There are many ways to regularize the theory. Among them we mention here the regularization of integration measure in \eqref{pf}, when $\hat{Q}$ remains chirally symmetric. In this approach the non-conservation of axial current is related to the non-invariance of the regularized integration measure over $\psi, \bar{\psi}$ with respect to the chiral transformation \cite{Fujikawa,Vergeles}. In the other regularization schemes the integration measure remains invariant under chiral transformation, but the regularized  operator $\hat{Q}$ is not invariant. This is the case of lattice regularization, for example, and in this approach the axial anomaly is related solely to the properties of the regularized $\hat{Q}$ (see \cite{KhaidukovZubkov}). In the present paper we will be using the point splitting regularization \cite{Schwinger}, which does not modify the integration measure, and also does not modify the Dirac operator $\hat{Q}$. Instead (see the next Section) the products of operators entering expressions for physical quantities are regularized: we avoid the ultraviolet divergencies by setting the splititng of point $z$ in the products $\bra{x}\hat{A}\hat{B}\ket{y} = \int dz \bra{x}\hat{A}\ket{z}\bra{z}\hat{B}\ket{y}$.  
\begin{align}
	\psi(x)\to e^{i\alpha(x)\gamma^5}\psi(x) \label{this} \\
	 \bar\psi(x) \to \bar\psi(x)e^{i\alpha(x)\gamma^5} \label{this1} 
\end{align}
The variation of $S^\epsilon$ to linear order in $\alpha$ is
\begin{align}
	\delta S^\epsilon=-i\text{tr}_{ D}\text{tr}_{ G}\text{tr}_{H}\left(\alpha(\hat x)\gamma^5\{\hat Q , \hat\rho^\epsilon \}\right) \label{fader}
\end{align}
Using \eqref{TrXY} we get
%Using the fact that $\text{tr}_{ G}\text{tr}_{H}(\hat X\hat Y)=\text{tr}_{G}\text{tr}_{\Gamma}(X_WY_W)$ (where $\text{tr}_\Gamma\equiv\int d^4x\int(2\pi)^{-4}d^4p$ is the trace w.r.t. the phase space), and that $(e^{i \hat{\pi} \epsilon})_W = e^{i p \epsilon} $, we get
\begin{align}
	\delta S^\epsilon&=\text{tr}_{G}\int d^4x\,\alpha(x)\Gamma^{\epsilon}(x)\label{a1}\\
	\text{where}\quad\Gamma^{\epsilon}(x)&=-i\text{tr}_{D}\gamma^5\int\frac{d^4p}{(2\pi)^{4}}(Q_W\bigstar\rho_W^\epsilon+\rho_W^\epsilon\bigstar Q_W) \label{ug}\\
    \text{and}\quad \rho_W^{\epsilon}&=e^{ip\epsilon}\bigstar\rho_W\bigstar e^{ip\epsilon}
\end{align}
Applying the definition of $\bigstar$ and integrating over $p$ we get
\begin{align}
	\begin{split}
		\Gamma^{\epsilon}(x)&=-i\text{tr}_{D}\gamma^5\int (2\pi)^{-8}d^4yd^4kd^4k'\,e^{-iy(k-k')}U(x,x+y/2)\\
		&\left(Q_W(x+y/2,k)\rho_W^\epsilon(x+y/2,k')+\rho_W^\epsilon(x+y/2,k)Q_W(x+y/2,k')\right)
		U(x+y/2,x) \label{geo}
	\end{split}
\end{align}
Using \eqref{adjoint} we can write this as
\begin{align}
	\Gamma^{\epsilon}(x)&=-i\text{tr}_{D}\gamma^5\int(2\pi)^{-8}d^4yd^4kd^4k'\,e^{-iy(k-k')}e^{y{\mathcal{D}}/2}\left(Q_W(x,k)\rho_W^{\epsilon}(x,k')+\rho_W^{\epsilon}(x,k)Q_W(x,k')\right) 
\end{align}
Integrating by parts we get
\begin{align}
	\Gamma^{\epsilon}(x)&=-i\text{tr}_{D}\gamma^5\int(2\pi)^{-8}d^4yd^4kd^4k'\,e^{-iy(k-k')}\left(e^{-i\partial_{k}{\mathcal{D}}/2}\left(Q_W(x,k)\rho_W^{\epsilon}(x,k')\right)+e^{i\partial_{k'}{\mathcal{D}}/2}\left(\rho_W^{\epsilon}(x,k)Q_W(x,k')\right)\right) \label{her}\\
    &\pd{=-i\text{tr}_{D}\gamma^5\int(2\pi)^{-4}d^4k\,\left(e^{-i\partial_{k}{\mathcal{D}}/2}\left(Q_W(x,k)\rho_W^{\epsilon}(x,k')\right)+e^{i\partial_{k'}{\mathcal{D}}/2}\left(\rho_W^{\epsilon}(x,k)Q_W(x,k')\right)\right)_{k'=k} }
\end{align}
%Expanding in powers of the covariant derivative and assuming that $Q(x,p)$ is such that 
%\begin{align}
%	\partial_{p_{\mu_1}}...\partial_{p_{\mu_n}}Q^{\epsilon}(x,p)=0\quad\text{for }n\geq 3
%\end{align}
%or 
%\begin{align}
%	\Big|\frac{1}{4^k (2k+1!)}{\mathcal{D}}_{\mu_1}...{\mathcal{D}}_{\mu_{2k}}	\partial_{p_{\mu_1}}...\partial_{p_{\mu_{2k}}}\partial_{p_{\mu_{0}}}Q^{\epsilon}(x,p) \Big| \ll \Big|\partial_{p_{\mu_{0}}}Q^{\epsilon}(x,p) \Big| \quad\text{for }k\geq 1 \label{ddQ}
%\end{align}
%we get
Expanding in powers of the covariant derivative and using $\{\gamma_5,Q_W\}=0$ we get
\begin{align}
	\Gamma^\epsilon(x)&={\mathcal{D}}_\mu J^{\epsilon}_\mu(x) \label{a2}\\
	\text{where}\quad J^{\epsilon}_{\mu}(x)&:=-\frac{1}{2}\text{tr}_{D}\gamma^5\int \frac{d^4p}{(2\pi)^4}\left(\partial_{p_\mu}Q_W(x,p)\rho_W^{\epsilon}(x,p)-\rho_W^{\epsilon}(x,p)\partial_{p_\mu}Q_W(x,p)\right)+...
\end{align}
where we have suppressed the higher order terms in the expansion. $J_\mu^\epsilon(x)$ is the \emph{chiral current}. Its conservation on-shell follows from the fact that when the equations of motion are satisfied, $\delta S^\epsilon=0$ for arbitrary $\alpha(x)$, and therefore (inspecting \eqref{a1}) $\Gamma^\epsilon(x)=\mathcal D_\mu J_\mu^\epsilon(x)=0$.
%The same expression appears if we assume that function $\alpha(x)$ varies slowly.

\section{The anomaly}

%From \eqref{a1},\eqref{a2} it is clear that \mzo{
%\begin{equation}
%     \text{tr}_G {\langle} \mathcal D_\mu J_\mu^\epsilon(x)\rangle=\bigg{\langle}\frac{\delta S^\epsilon}{\delta \alpha(x)}\bigg{\rangle}
%\end{equation}}
\pd{From \eqref{a2} we get
\begin{equation}
     \text{tr}_G {\langle} \mathcal D_\mu J_\mu^\epsilon(x)\rangle=\tr_G\langle\Gamma^{\epsilon}(x)\rangle
\end{equation}}
where $\langle\cdot\rangle$ denotes the expectation value. 
Since $\hat Q$ is the operator governing the quadratic term in the action it is clear that $\langle\hat\rho\rangle=\mzo{-}\hat Q^{-1}$. Let us denote this as $\hat G:=\hat Q^{-1}$. \pd{From \eqref{ug}, then, we get (using $\tr \gamma_5=0$)}
%Denoting $\langle \rho_W\rangle=G_W$ we get
\mzo{\begin{align}
    \text{tr}_G \langle \mathcal D_\mu J_\mu^\epsilon\rangle=&
    i\text{tr}_{ D}\text{tr}_{G}\gamma_5\int (2\pi)^{-4}d^4p \left(Q_W\bigstar e^{ip\epsilon}\bigstar G_W\bigstar e^{ip\epsilon} + e^{ip\epsilon}\bigstar  G_W \bigstar e^{ip\epsilon}\bigstar Q_W\right)
    \\
    =&i\text{tr}_{ D}\text{tr}_{G}\gamma_5\int (2\pi)^{-4}d^4p \left( (Q_W\bigstar e^{ip\epsilon}-e^{ip\epsilon}\bigstar Q_W) \bigstar G_W \bigstar e^{ip\epsilon} -e^{ip\epsilon} \bigstar G_W \bigstar (Q_W\bigstar e^{ip\epsilon}-e^{ip\epsilon}\bigstar Q_W) \right) \label{last1}
\end{align}}
\mza{If we add the integration over $x$ to this expression, then using cyclic property of the functional trace and chiral symmetry of $Q_W$, we will arrive at a vanishing answer for the global chiral anomaly. Integration over coordinate space makes the total expression divergent, and we need a regularization in order to fix the problem. We choose a kind of infrared regularization, in which the integration over infinite coordinate space is replaced by the integration over a finite space with boundary. It is assumed that this finite coordinate space is sufficiently large so that we may still think of momenta as continuous.
This regularization does not allow us to use identities (\ref{cyclTr}). As a result the expression for global chiral anomaly remains finite and nonzero. We require that at the boundary of coordinate space the field strength of the external gauge field is vanishing. We also require that the other sources of inhomogeneity that may be present in the considered system disappear at infinity, which effectively means that spatial derivatives of $Q_W$ tend to zero at the boundary of coordinate space.}

We wish to expand \eqref{last1} \mzz{to  quadratic order in the field strength. Expansion of the star in powers of $F$ will result in the terms $\sim e^{2 i\epsilon p} \epsilon^n F^m$ with $m\ge n$. We will see below that such terms with $n>1$ result in vanishing contributions to anomaly at $\epsilon \to 0$. Therefore, we restrict ourselves by the terms with $n = 0, 1$. }
Using the definition of $\bigstar$ it can be shown that
\begin{align}
	e^{ip\epsilon}\bigstar Q_W-Q_W\bigstar e^{ip\epsilon}=&
	e^{ip\epsilon}(2\pi)^{-4}\int d^4yd^4k\,e^{-iy(k-p)}\left(\mathscr W(x,\epsilon,y)Q_W(x+\epsilon/2,k)-\mathscr W(x,y,\epsilon)Q_W(x-\epsilon/2,k)\right)\label{com}\\
	\begin{split}
		\text{where}\quad \mathscr W(x,y,y'):=& U(x , x-(y+y')/2) U(x-(y+y')/2 , x-y'/2)  U(x-y'/2 , x+(y-y')/2)\times\\
		&U(x+(y-y')/2, x+y/2)  U(x+y/2,  x+(y+y')/2)  U(x+(y+y')/2 , x)
	\end{split}
\end{align}
is the triangular Wilson loop in Fig. \ref{mfig} without the insertions. 
\pd{
In a previous paper \cite{zubkovxavier} we reported the following expansion of $\mathscr{W}$ to order $\mathcal D^4$:
\begin{align}
\begin{split}
\mathscr{W}(x,y,y')=&1-\frac{i}{2}y^\mu y'^\nu F_{\mu\nu}(x)-\frac{i}{12}(y^\alpha-y'^\alpha)y^\mu y'^\nu \mathcal{D}_\alpha F_{\mu\nu}(x)-\frac{i}{48}(y^\alpha y^\beta+y'^\alpha y'^\beta) y^\mu y'^\nu \mathcal{D}_{\alpha}\mathcal{D}_\beta F_{\mu\nu}(x)\\
   &-\frac{1}{8}y^\mu y'^\nu y^\alpha y'^\beta F_{\mu\nu}(x)F_{\alpha\beta}(x)+\mathcal O(\mathcal D^5)
\end{split}
\end{align}
Expanding \eqref{com} to order $\epsilon$ and $O(\mathcal D^4)$ using these we get
\begin{align}
    e^{ip\epsilon}\bigstar Q_W-Q_W\bigstar e^{ip\epsilon}=&
    e^{ip\epsilon}\left(\epsilon_\mu\left(\partial_{x_\mu}Q_W-F_{\mu\nu}\partial_{p_\nu}Q_W+\frac{1}{24}\mathcal D_\alpha\mathcal D_\beta F_{\mu\nu}\partial_{p_\alpha}\partial_{p_\beta}\partial_{p_\nu}Q_W+O(\mathcal D^5)\right)+O(\epsilon^2)\right)
\end{align}
}
%Obviously $\mathscr W$ has the following expansion:
%\begin{align}
%	\mathscr{W}(x,y,y')=&1-\frac{i}{2}y^\mu y'^\nu F_{\mu\nu}(x)+O(y^2)+O(y'^2)
%\end{align}
%Using this we can expand \eqref{com} to order $\epsilon$ and linear order in the field strength:
%\mzo{\begin{align}
%	e^{ip\epsilon}\bigstar Q_W-Q_W\bigstar e^{ip\epsilon}=&
%	e^{ip\epsilon}\left(\epsilon_\mu\left(\partial_{x_\mu}Q-F_{\mu\nu}\partial_{p_\nu}Q\right)+O(F^2\epsilon^2)\right) \label{epsi1}
%\end{align}}
%where $O(\mathcal D^3)$ signifies terms arising from 3 or more powers of $\mathcal D$ -- in this sense $F_{\mu\nu}$ is $O(\mathcal D^2)$.
%Since \eqref{epsi1} is already of order $\epsilon$, we may set $\epsilon=0$ in 
\pd{
Using the identity $\text{tr}_G\int d^4p\, X_W\bigstar Y_W=\int d^4p\left(e^{-i\partial_p\partial_x/2}(X_W(x,p)Y_W(x,k))\right)\big{|}_{k=p}$ (which follows from \eqref{starer}) in \eqref{last1} we obtain
\begin{align}
	 &\int d^4x\, \text{tr}_G \langle \mathcal D_\mu J_\mu^\epsilon\rangle
	=-i \int d^4x\, \text{tr}_{ D}\text{tr}_{G}\gamma_5\int (2\pi)^{-4}d^4p\, \epsilon_\mu\times\\
    &\Bigl(\Bigl[ e^{ip\epsilon}(\partial_{x_\mu}Q_W-F_{\mu\nu}\partial_{p_\nu}Q_W+\frac{1}{24}\mathcal D_\alpha\mathcal D_\beta F_{\mu\nu}\partial_{p_\alpha}\partial_{p_\beta}\partial_{p_\nu}Q_W)\Bigr] \,\mza{e^{-i\Ll\partial_p (\overleftarrow{\partial}_x + \overrightarrow{\partial}_x )/2}}\,  \Bigl[G_W \bigstar e^{ip\epsilon}\Bigr]\\ 
    & -\Bigl[e^{ip\epsilon} \bigstar G_W\Bigr]\,\mza{e^{-i\Ll\partial_p (\overleftarrow{\partial}_x + \overrightarrow{\partial}_x )/2}}\, \Bigl[e^{ip\epsilon} (\partial_{x_\mu}Q_W-F_{\mu\nu}\partial_{p_\nu}Q_W+\frac{1}{24}\mathcal D_\alpha\mathcal D_\beta F_{\mu\nu}\partial_{p_\alpha}\partial_{p_\beta}\partial_{p_\nu}Q_W)\Bigr] \Bigr) \\
    &=-i \int d^4x\,\text{tr}_{ D}\text{tr}_{G}\gamma_5\int (2\pi)^{-4}d^4p \,\epsilon_\mu e^{ip\epsilon}(\partial_{x_\mu}Q_W-F_{\mu\nu}\partial_{p_\nu}Q_W+\frac{1}{24}\mathcal D_\alpha\mathcal D_\beta F_{\mu\nu}\partial_{p_\alpha}\partial_{p_\beta}\partial_{p_\nu}Q_W)(G_W\bigstar e^{ip\epsilon}+e^{ip\epsilon}\bigstar G_W)\label{last2} \\&\quad + \mza{{\rm boundary}\,{\rm terms}}
\end{align}
\mza{Here by boundary terms we understand the integral over the boundary of coordinate space of the terms containing spatial derivatives of $Q_W$ and $G_W$ as well as the field strength $F$. These latter terms vanish according to our supposition.}
\pd{
Now
\begin{align}
	e^{ip\epsilon}\bigstar G_W+G_W\bigstar e^{ip\epsilon}
	=& e^{i p \epsilon}(2G_W + O(\epsilon))
\end{align}
(it is only necessary to expand to order $\epsilon^0$ here since \eqref{last2} is already proportional to $\epsilon$). So we get
}
\begin{equation}
    \int d^4x\, \text{tr}_G \langle \mathcal D_\mu J_\mu^\epsilon\rangle=-2i\,\text{tr}_{ D}\text{tr}_{G}\gamma_5\int (2\pi)^{-4}d^4xd^4p \,e^{2ip\epsilon} \epsilon_\mu \left(\partial_{x_\mu}Q_W-F_{\mu\nu}\partial_{p_\nu}Q_W+\frac{1}{24}\mathcal D_\alpha\mathcal D_\beta F_{\mu\nu}\partial_{p_\alpha}\partial_{p_\beta}\partial_{p_\nu}Q_W\right)G_W\label{fern}
\end{equation}
}
%\mzo{At the end of the calculation of the anomaly we will take off both the infrared regularization by an integration over finite coordinate space and the ultraviolet regularization with finite $\epsilon$.}

\begin{theorem} \label{thm1} For \mzo{a sufficiently smooth} function $f(p)$
    \begin{equation}
        \lim_{|\epsilon|\to 0} \left\langle\int d^4p\,e^{ip\epsilon} \epsilon_\mu f(p)\right\rangle=i\int d^4p\,\partial_\mu f(p) \label{thm}
    \end{equation}
    where $\langle\cdot\rangle$ denotes averaging over the directions of $\epsilon_\mu$ (i.e. for fixed $|\epsilon|$ integrate over the 3-sphere of radius $|\epsilon|$ and divide by the volume of the integration region.)
\end{theorem}
\begin{proof}
Let us denote  by ${\rm FT}(f) = \int \frac{d^4 p}{(2\pi)^4} e^{ip\epsilon} f(p)$ the Fourier transform. And let us denote the inverse Fourier transform by $\rm IFT$.  	
Now $f=\text{IFT}(\text{FT}(f))\implies \partial_\mu f=\text{IFT}(-i\epsilon_\mu\text{FT}(f))\implies \text{FT}(i\partial_\mu f)=\epsilon_\mu\text{FT}(f)\implies \langle\text{FT}(i\partial_\mu f)\rangle=\langle\epsilon_\mu\text{FT}(f)\rangle$. Taking the limit $|\epsilon|\to 0$ on both sides gives \eqref{thm}. Henceforth, for ease, we will denote the process of averaging over directions of $\epsilon_\mu$ and taking $|\epsilon|\to 0$, collectively, as $\lim_{|\epsilon|\to 0}$.
\end{proof}
Applying this to \eqref{fern} we get
\pd{\begin{align}
    \mathscr A:=&\mzo{}\lim_{|\epsilon|\to 0}\int d^4x\, \text{tr}_G \langle \mathcal D_\mu J_\mu^\epsilon\rangle=\mzo{+}\text{tr}_{ D}\text{tr}_{G}\gamma_5\int (2\pi)^{-4}d^4xd^4p\, \partial_{p_\mu} \left(\left(\partial_{x_\mu}Q_W-F_{\mu\nu}\partial_{p_\nu}Q_W+\frac{1}{24}\mathcal D_\alpha\mathcal D_\beta F_{\mu\nu}\partial_{p_\alpha}\partial_{p_\beta}\partial_{p_\nu}Q_W\right)G_W\right) \label{fern1}
\end{align}}
$\mathscr A$ is the anomaly. 
%In the second line above we have used the divergence theorem to write the integral of the total momentum derivative as an integral over the 3-sphere at infinity.

\subsection{\texorpdfstring{Smooth deformation of the system removing $x$ dependence.}{}}

Let us assume that $Q_W(x,p)$ is homotopic to a function $\widetilde Q(p)$. Since $\mathscr A$ is the topological index of the relevant operator (as discussed in the introduction) \eqref{fern1} should be invariant if we replace $Q_W(x,p)$ with $\W{Q}(p)$ (since this amounts to a smooth deformation of the relevant operator):
\pd{\begin{equation}
   \mathscr A=\mzo{-\text{tr}_{ D}}\text{tr}_{G}\gamma_5\int (2\pi)^{-4}d^4xd^4p\, \partial_{p_\mu} \left((F_{\mu\nu}\partial_{p_\nu}\W{Q}-\frac{1}{24}\mathcal D_\alpha\mathcal D_\beta F_{\mu\nu}\partial_{p_\alpha}\partial_{p_\beta}\partial_{p_\nu}\W{Q})\W{G}_W\right) \label{ind1}
\end{equation}}
where $\widetilde G_W$ is the Greens function derived from $\widetilde Q$, i.e. it satisfies $\W{G}_W\bigstar \W{Q}=1$.

%Let us proceed to calculate with $\widetilde Q$. At the end we will return to an expression with $Q(x,p)$.

We now make use of a result that we derived in a previous paper \cite{zubkovxavier}: $\widetilde G_W$ has an expansion in powers of the field strength:
\begin{align}
    \widetilde G_W=&\widetilde G^{(0)}+\frac{i}{2}\widetilde G^{(0)} \partial_{p_{\alpha}}\W{Q}\, \widetilde G^{(0)}  \partial_{p_\beta} \widetilde Q\,\widetilde G^{(0)} F_{\alpha\beta}+O(F^2)
\end{align}
%where $\star$ is the Moyal star product introduced in \eqref{moyalstar} and 
where $\widetilde G^{(0)}$ is the solution to $\widetilde G^{(0)} \widetilde Q=1$.
We substitute this into \eqref{ind1} and obtain 
\begin{align}
    \mathscr A=&\mzo{-\int d^4x} \, \tr(F_{\mu\nu})\int (2\pi)^{-4}d^4p\,\text{tr}_{ D}(\gamma_5\partial_{p_\nu}\widetilde Q\partial_{p_\mu}\widetilde G^{(0)})  \label{firsty}  \\
    &\mzo{+}\frac{i}{2}\int d^4x\, \text{tr}(F_{\mu\nu}F_{\alpha\beta})\int (2\pi)^{-4}d^4p\,   \partial_{p_\mu}S_{\alpha\beta\nu}\\
    &\pd{+\frac{1}{24}\int d^4x\, \partial_\alpha\partial_\beta\tr(F_{\mu\nu})\int (2\pi)^{-4}d^4p\tr_D(\gamma_5\partial_{p_{\alpha}}\partial_{p_{\beta}}\partial_{p_{\nu}}\W{Q}\partial_{p_{\mu}}\W{G}^{(0)})}\label{reni}
    \\
    \text{where}\quad S_{\alpha\beta\nu}(x):=&\frac{1}{2}\text{tr}_D\left(\gamma^5 \widetilde G^{(0)}\partial_{p_\alpha} \widetilde Q \,\widetilde G^{(0)}\partial_{p_\beta} \widetilde Q\, \widetilde G^{(0)}\partial_{p_\nu} \widetilde Q\right)-(\alpha\leftrightarrow \beta)\label{S3}
\end{align}
%If $G$ is a non-Abelian gauge group then $\tr F_{\mu\nu}=0$, but since we want our derivation to include the case of Abelian $G$ we will not make use of this identity. Instead, 
Now, consider 
\begin{align}
    \tr_{ D}(\gamma_5\partial_{p_\nu}\widetilde Q\partial_{p_\mu}\widetilde G^{(0)})=&\text{tr}_{ D}(\gamma_5\partial_{p_\nu}\widetilde Q\W{G}^{(0)}\W{Q}\partial_{p_\mu}\widetilde G^{(0)})\\
    =&-\text{tr}_{ D}(\gamma_5\widetilde Q\partial_{p_\nu}\W{G}^{(0)}\W{Q}\partial_{p_\mu}\widetilde G^{(0)})\\
    =&-\text{tr}_{ D}(\gamma_5\widetilde Q\partial_{p_\mu}\W{G}^{(0)}\W{Q}\partial_{p_\nu}\widetilde G^{(0)}) 
\end{align}
where in the first line we've inserted $1=\W{Q}\W{G}^{(0)}$, in the second line we've used $\partial_{p_\nu}\W{Q}\W{G}^{(0)}=-\W{Q}\partial_{p_\nu}\W{G}^{(0)}$ and in the last line we've used the cyclic property of trace and $\{\gamma_5,\W{Q}\}=\{\gamma_5,\W{G}^{(0)}\}=0$. The last two lines show that the expression is $\mu\leftrightarrow\nu$ symmetric. Multiplication by $F_{\mu\nu}$ in \eqref{firsty} then kills this term.
Now, using similar arguments, it is easy to verify that $S_{\alpha\beta\nu}$ is totally antisymmetric. Therefore it is a 3-form in 4 dimensional momentum space.
%Therefore it can be written as
%\begin{align}
%	S_{\alpha\beta\nu}=&\varepsilon_{\alpha\beta\nu\sigma}S^\star_\sigma\\
%	\text{where}\quad S^\star_\sigma :=&\varepsilon_{\alpha\beta\nu\sigma}S_{\alpha\beta\nu}/3! \label{s3}
%\end{align}	
%where $S^\star_\mu$ is the Hodge dual.
Now, it is easy to verify that $\text{tr}(F_{\mu[\nu}F_{\alpha\beta]})$ is totally antisymmetric and therefore $\text{tr}(F_{\mu[\nu}F_{\alpha\beta]})=\varepsilon_{\mu\nu\alpha\beta}\tr(FF^\star)/12$ where \mzo{$F^\star_{\mu\nu} =\frac{1}{2} \epsilon_{\mu\nu\rho\sigma}F_{\rho\sigma}$} is the Hodge dual.
%Furthermore, one can verify that $\text{tr}\left(F_{\mu\nu}F_{\alpha\beta}\right)\varepsilon_{\alpha\beta\nu\sigma} =-\frac{1}{2} \text{tr}(F_{\alpha\beta} F^\star_{\alpha\beta}) \delta_{\mu\sigma}$ where $ F^\star_{\mu\nu}:=\varepsilon_{\mu\nu\alpha\beta} F_{\alpha\beta}/2$ is the Hodge dual.
%Furthermore, since $S_{\alpha\beta\nu}$ us independent of $x$ we can introduce an integral over $x$ at no cost, provided we divide by the ``4-volume''. 
\pd{Now \eqref{reni} vanishes since we assume the field strength is zero on the boundary.}
So we arrive at
\begin{align}
	\mathscr A=
    &\mzo{-2i N_3} \int \frac{1}{16\pi^2}d^4x\, \text{tr}(FF^\star)\label{index05}\\
    \text{where}\quad N_3=&\mzo{\frac{1}{8\pi^2}}\int  dS
\end{align}
($S$ is a 3-form so $dS$ is a top-form which is integrated over 4 dimensional momentum space.)
%At this point let us introduce the infrared regularization, deforming $Q(p) \to Q_m(p)$  ($G^{(0)}\to Q^{-1}_m(p)$) in such a way that the pole of this function at finite position of momentum is removed. Such a regularization necessarily breaks chiral symmetry close to the would be Fermi point. (The standard way of introducing the infrared regularization for the conventional Dirac operator is turning on finite mass $m$.) The advantage of having the expression regularized in the infrared is that the following theorem holds.
%The infrared regularization is needed in order for the Fourier transform of $\partial_\mu S^{(m)}_{\alpha\beta\nu}(p)$ to be defined, at least as a generalized function of $x$. Without it, for example, for the conventional Dirac fermions with $Q = \gamma p$ the integral over $p$ in the Fourier transform is divergent at $p \to 0$. Notice, that the same integral over $p$ might be divergent in ultraviolet, at $p \to \infty$. However, then the resulting Fourier transform is defined as the generalized function of $x$, which may be proportional to the derivatives of delta function $\delta(|x|)$.  
$S$ may have a pole at $p=0$ so the integral may not be defined. If we introduce a small mass into the system, however, the pole moves into the complex plane. Let us assume this has been done so that we may disregard the pole at $p=0$ and then the small mass taken to zero. 
Then Stokes' theorem gives us
\begin{align}
	N_3
	=& \mzo{\frac{1}{8\pi^2}}\int_{S^3_{\infty}}  S
\end{align}
where $S^3_\infty$ is the 3-sphere at infinity, which we treat as the boundary of momentum space.
%the integration is over the 3-sphere at infinity with normal integration element $d\sigma_\mu$. 
We can deform this sphere at infinity to the 3-surface $\Sigma$ defined as the union of the two hyperplanes $p_4=0^{\pm}$ (whose orientation is defined by the normal vector $\sgn(p_4)\partial_{p_4}$) slightly above and below $p_4=0$. Then we get
\begin{align}
	N_3
	=& \mzo{\frac{1}{8\pi^2}}\int_{\Sigma}  S=\mzo{\frac{1}{48\pi^2}}\int _\Sigma\text{tr}_D\left(\gamma^5 \widetilde G^{(0)}d \widetilde Q \,\widetilde G^{(0)}\wedge d \widetilde Q \widetilde G^{(0)}\wedge d \widetilde Q\right) \label{fer1}
\end{align}
Let $\star=e^{\frac{i}{2}(\Ll{\partial}_x\Rr\partial_p-\Ll\partial_p\Rr\partial_x)}$ denote the Moyal star product \cite{zachos2005quantum, moyal1949quantum}. Since there is no $x$ dependence in \eqref{fer1} we can insert $\star$ for free. We can also insert $1=\frac{1}{|V|}\int d^3\vec x$ where $|V|$ is the volume of space. So we can write \eqref{fer1} as
\begin{align}
	N_3
	=\mzo{\frac{1}{48\pi^2|V|}}\int d^3\vec x\int _\Sigma\text{tr}_D\left(\gamma^5 \widetilde G^{(0)}\star d \widetilde Q\star \widetilde G^{(0)}\star\wedge d \widetilde Q \star\widetilde G^{(0)}\star\wedge d \widetilde Q\right) \label{fer11}
\end{align}
Now we observe that \eqref{fer11} is a topological invariant: i.e. it is invariant if we replace $\W{Q}(p)$ with the original $Q_W(x,p)$ to which it is homotopic. Similarly, we have to replace $\W{G}^{(0)}$ with $G^{(0)}$ which is the solution to $Q_W(x,p)\star G^{(0)}=1$. So we can write \eqref{fer11} as
\begin{align}
	N_3
	=&\mzo{\frac{1}{48\pi^2|V|}}\int d^3\vec x\int _\Sigma\text{tr}_D\left(\gamma^5  G^{(0)}\star d Q_W\star  G^{(0)}\star\wedge d  Q_W \star G^{(0)}\star\wedge d  Q_W\right)\label{der1}
\end{align}
%\\
%    =&\sum_{p_4=0^{\pm}}\frac{\sgn(p_4)}{48\pi^2V_4}\int d^4x\int d^3\vec p\,\varepsilon_{ijk}\,\text{tr}_D\left(\gamma^5  G^{(0)}\star \partial_i Q_W\star  G^{(0)}\star \partial_j  Q_W \star G^{(0)}\star \partial_k  Q_W\right)
\emph{We conclude finally} that the chiral anomaly $\mathscr A$ for $\hat Q$ is given by 
\begin{equation}
    \mathscr A=\mzo{-2i N_3} \int \frac{1}{16\pi^2}d^4x\, \text{tr}(FF^\star) \label{ganm}
\end{equation}
with $N_3$ given by \eqref{der1}.

We note that if we make an inverse Wick rotation back to Minkowski spacetime then we obtain
\begin{equation}
    \mathscr A=\mzo{}N_3\times \frac{1}{2\pi^2}\int d^4x\,\text{tr}({\bf E}.{\bf B})
\end{equation}

\section{Illustrative example}
\label{SectIllustr}

To illustrate the formula derived above let us consider a system of Dirac fermions with (as it will be shown) $N_3\neq 1$ (ordinary Dirac fermions have $N_3=1$ as shown, for example, in \cite{SuleymanovZubkov2020}). 
Consider the system with
\begin{align}
	\hat{Q} =&  \mzc{\left(\begin{array}{cc}0 & \hat O^\dagger\\\hat O&0 \end{array}\right)}\\
    \text{where}\quad 
	 \hat O=& \hat{\pi}_4 \pd{+} i   \left( \begin{array}{cc}\hat{\pi}_3 & \kappa(\hat{\pi}_1-i\hat{\pi}_2)^n \\\kappa(\hat{\pi}_1+i\hat{\pi}_2)^n&-\hat{\pi}_3 \end{array}\right)
\end{align}
with integer $n$, constant $\kappa$ and
$\hat{\pi}_\mu := \hat{p}_\mu - A_\mu(\hat x)$ (see \eqref{pi}).
%with $\hat{p} = - i \partial$, we denote  the $4$ - momentum. One can check that in this model $N_3 = n$, when the integration in momentum space is along the three-dimensional hypersurface surrounding the Fermi-point $p = 0$. 
(Note that ordinary Dirac fermions correspond to the case $\kappa=n=1$.)
In this model
\begin{align}
	S_{\alpha\beta\nu}(p)=&\pd{+}\varepsilon_{\alpha\beta\nu\sigma}\tilde p_\sigma\frac{4 \kappa ^2 n^2 \left(p_1^2+p_2^2\right)^{n-1}}{\left(\kappa ^2
		\left(p_1^2+p_2^2\right)^n+p_3^2+p_4^2\right)^2}\label{S1}\\
	\text{where}\quad \tilde p_\sigma:=&(p_1/n,p_2/n,p_3,p_4)
\end{align}
(see \eqref{S3}) and it can be shown that
\begin{equation}
    \int dp_1dp_2dp_3 \frac{1}{8\pi^2} S_{123}(p_1,p_2,p_3,p_4)=\pd{+}\frac{n}{2}\sgn(p_4)
\end{equation}
Using \eqref{fer1} then gives $N_3=n$.
%(the reason the minus doesn't survive is because the induced volume form is given by $i_{\partial_4}(d^4p)=-dp_1\wedge dp_2\wedge dp_3$ \cite{eggers_boundary_orientation} where $i_{(\cdot)}$ is the ``interior product'').
Hence the formula for the anomaly here is
\begin{equation}
    \mathscr A=\mzo{-n}\times \frac{i}{4\pi^2}\int \tr (F\wedge F)
\end{equation}

\section{Conclusion and discussion}

In this paper we considered a noninteracting model of fermions in the presence of an external gauge field (which is, in general, non-Abelian). The system under consideration may possess nontrivial internal momentum-space topology, arising from the nontrivial dependence of the generalized Dirac operator on momenta. Moreover, we considered the case where these systems are inhomogeneous even if the external gauge field is removed. We required, however, that this additional inhomogeneity can be eliminated through a smooth modification of the system.

Under these conditions, we found that the anomaly arising from the generalized Dirac operator factorizes into the product of the topological charge carried by the external gauge field (equal to the number of instantons) and the topological invariant $N_3$, which is responsible for the stability of the Fermi surfaces/Fermi points.
%, see \eqref{der1}, \eqref{ganm}. 
Notably, $N_3$ also appears in the expression for the conductivity of the chiral separation effect (both Abelian and non-Abelian).

\mzu{The approximation we are considering is the “derivative expansion” in which the gauge field $A_\mu$ and partial derivatives $\partial_\mu$ are considered "order 1": for example $\partial_\mu A_\nu$ and $A_\mu A_\nu$ are both of “order 2”. We refered to this as to expansion in powers of the covariant derivative $D_\mu \equiv \partial_\mu-iA_\mu$. Our result is valid to order $D^4$, i.e. to quadratic order in the field strength. Beyond this order we do not exclude the appearance of the contributions like $D_\alpha F_{\mu\nu} F_{\rho\sigma}$ or $F^3$. Notice that these extra terms do not appear in ordinary relativistic quantum field theory with the conventional Dirac operator linear in momentum. Moreover, one may expect that they will disappear order by order in the considered models as well. This hope is based on the observation that we actually deal with a version of Atiyah – Singer theorem. The right hand side in our case is given by the product of phase space topological invariant and the number of instantons. There is a limited number of topological invariants, and in the situation when the two topologies (that of momentum space and that of the gauge field) are factorized, presumably, the expression obtained by us is the only choice. However, this issue is out of the scope of the present paper.}

\mzu{Since we are considering the models with a rather general form of Dirac operator, the possible problem of the ghost states is worth to be commented. There are basically two ways to look at this question. The first one is based on the experience of relativistic quantum field theory. The presence of the higher derivative terms in the classical action (typically, the derivatives are considered up to the fourth order) cause problems in canonical quantization related to the appearance of indefinite metric in Hilbert space to be constructed (see, for example, \cite{SMILGA2005598} and references therein). The second way is based more on the experience of condensed matter physics theory. There the one – particle Hamiltonian that may be an arbitrary function of momentum is a conventional rather than marginal situation. Say, in the tight binding models of solid state physics the terms proportional to cosine or sine of momentum are typical. The second quantization based on the one particle Hilbert space does not encounter any difficulties: the complete field theory Hilbert space does not suffer from the problems with indefinite metrics \cite{altland2010condensed}. In the present paper we rely on this approach, and assume that the relativistic field theory appears as a certain limiting case of a condensed – matter like construction \cite{volovik2003universe}. }

It would be interesting to extend this research in several directions. In particular, the role of interactions may be investigated. 
\mzu{Their consideration is assumed to become the subject of a future work. Then the violation of chiral symmetry due to the chiral anomaly together with the absence of scale symmetry might cause the chiral symmetry breakdown at the level of the effective Dirac operator (that is inverse to the two point Green function with interaction corrections). At low energies the mass term will appear, i.e. the term with the transition between the left – handed and the right – handed fermions. If the non – Abelian gauge field entering expression for the chiral anomaly is dynamical, then this mass term is suppressed when the instanton density is suppressed. For the gauge field with the standard Yang – Mills action this occurs in the weak coupling regime.  Then the chiral symmetry is still present on the level of effective action. We expect \cite{zubkov2023effect}, that in this situation the result of the present paper should survive in some form (in the presence of additional non - dynamical external gauge field).} In the opposite case of strong coupling Eq. (\ref{ganm}) may be used, when we consider fermions in the presence of the given fixed configuration of the dynamical gauge field. (The divergence of axial current then becomes a dynamical variable as well as the gauge field.) Then Eq. (\ref{ganm}) may even become useful in consideration of various non - perturbative effects, say, being unified with the technique of central or Abelian projection (see, for example,  \cite{bakker1999central} and references therein).

 Furthermore, it would be worthwhile to consider cases where the inhomogeneity cannot be removed through smooth deformations. We expect that in such situations, \eqref{ganm} may still hold in some form under less restrictive conditions than those considered above. However, a comprehensive analysis of the conditions required for the continued validity of \eqref{ganm} lies beyond the scope of the present paper.  \mzu{In this respect it is worth mentioning that the topological invariant of the form of \eqref{der1} is known in the noncommutative geometry as pairing of the Connes - Chern character with a cyclic cocycle of the cyclic cohomology theory developed by Alain Connes \cite{Connes,khalkhali2013basic} (the other version of this formalism may be found in \cite{fedosov1996deformation}). This theory may, presumably, be used for  the calculation of the value of $N_3$ in rather general cases. More simple technique that may be applied to a certain class of such systems has been developed for such a calculation recently in \cite{Z2025}.} 

Another interesting question concerns the possible observation of phenomena, connected to the chiral anomaly, that can reflect possibly non-trivial (i.e. $N_3>1$) values of $N_3$ appearing in \eqref{ganm} (arising not simply from an increased number of fermion fields but rather from nontrivial momentum-space topology -- as in the example in \S\ref{SectIllustr}). 
It is possible that quark matter under exceptional external conditions (such as strong magnetic fields, high temperature, or pressure) may exhibit nontrivial values of $N_3$ \cite{zubkov2018momentum}. \mzu{This concerns also the Standard Model in general \cite{zubkov2012momentum,volovik2017standard}. Notice that the generalizations of the Standard Model, say with the composite Higgs bosons are related to condensed matter theory, and in particular to the theory of fermionic superfluids \cite{volovik2013nambu}.}
%In such a case, the chiral anomaly with unusual coefficient could lead to observable physical effects.
A similar situation may, in principle, also occur in solid-state systems that simulate high-energy physics in the laboratory; In some Weyl or Dirac semimetals, ordinary Weyl or Dirac points may merge, forming systems with $N_3 > 1$ (see the extensive discussion of this issue in \cite{volovik2003universe} and references therein). 
%In such systems, observable phenomena may exist that are intimately connected to the chiral anomaly with an unconventional coefficient in front of the topological density of the gauge field.

Finally, we would like to mention the potential possibility to extend our results to the case of the other currents, which are conserved classically. In particular, for the non - interacting systems with relativistic symmetry one can consider the currents generated by zilch fields \cite{kibble1965conservation,lipkin1964existence,huang2020zilch,chernodub2018zilch}. In case of the more complicated systems under certain conditions multiple classically conserved currents may be defined as well \cite{konstein2001conformal}. Some of the currents associated with zilch, as well as the other classically conserved currents may be anomalous \cite{copetti2018higher,artem2021zilch}. Correspondingly, the appearance of an analogue of our Eq. (\ref{ganm}) might be expected. However,  these  (potentially anomalous) currents  contain the algebraic structures that are not present in the ordinary axial current. Presumably, such structures are to replace $\gamma^5$ in Eq. (\ref{der1}). It is not clear whether the resulting expression in general case will remain a topological invariant. It could be, however, that it will remain robust to the smooth modifications of the system that respect symmetries resulted in the appearance of the corresponding classically conserved current. The consideration of this issue, however, is far out of the scope of the present paper.

\appendix

\bibliography{biblio_corrected}

\end{document}